\documentclass[12pt,a4paper, twoside,reqno]{amsart}
\usepackage{amscd}
\usepackage{amsmath,amssymb,amsthm,mathrsfs}
\usepackage[colorlinks=true,linkcolor=blue,citecolor=blue]{hyperref}
\usepackage[margin=3cm]{geometry}
\usepackage{tikz}
\usetikzlibrary{calc}
\usepackage{parskip}
\usepackage{enumerate}
\newtheorem{thm}{Theorem}[section]
\newtheorem{prop}[thm]{Proposition}

\theoremstyle{definition}
\newtheorem{definition}[thm]{Definition}

\newcommand{\e}{\epsilon}
\renewcommand{\L}{\mathfrak{L}}

\theoremstyle{remark}
\newtheorem{remark}[thm]{Remark}
\numberwithin{equation}{section}

\begin{document}
\begin{abstract}
We study combinatorial Laplacians on rectangular subgraphs of $ \epsilon \mathbb{Z}^2 $ that approximate Laplace-Beltrami operators of Riemannian metrics as $ \epsilon \rightarrow 0 $. These laplacians arise as follows: we define the notion of a Riemmanian metric structure on a graph. We then define combinatorial free field theories and describe how these can be regarded as finite dimensional approximations of scalar field theory. We focus on the Gaussian field theory on rectangular subgraphs of $ \mathbb{Z}^2 $ and study its partition function by computing the asymptotic determinant of the discrete laplacian.
\end{abstract}
\title{Asymptotic Determinant of Discrete Laplace-Beltrami Operators}
\author{Ananth Sridhar}
\address{Department of Physics, University of California, Berkeley, CA
94720, USA}
\maketitle
\section{Introduction}
On a graph $ X = (V,E) $, the graph laplacian is an operator that acts on the space of functions $ f: V \rightarrow \mathbb{R} $ as:
\begin{align} \label{eq:eq1}
\Delta f (u) = \sum_{ [u,v] \in E} f(v) - f(u) 
\end{align}
The determinants of graph laplacians have been well studied in many different contexts. For instance, in spectral graph theory, Kirchoff's theorem relates the determinant of the graph laplacian on $ X $ to the number of rooted spanning trees on $ X $. Similarly, by Temperley's trick and its generalizations, the determinant of the graph laplacian on certain graphs $ X $ computes the number of perfect matchings on a related graph $ X' $  \cite{KPW}. In each case, the determinant of the laplacian is the partition function of an associated model in statistical mechanics. 

When $ X $ is the graph of a regular lattice, the asymptotics of the determinant have been investigated in many papers. In particular, when $ X $ is a rectangular subgraph of $ \mathbb{Z}^2 $ with $ N \times M $ vertices, the determinant of the laplacian (with Dirichlet boundary conditions) as $ N,M \rightarrow \infty $ is computed by \cite{duplant} in the context of Hamiltonian walks on the grid. The asymptotic expansion is:
\begin{align} \label{eq:eq2}
\log \det \Delta_{NM} = N M \frac{ 4 G }{\pi} + 2 (N+M) \log(\sqrt{2}-1) + \frac{1}{4} \log(NM) + O(1)
\end{align}
where $ G $ is the Catalan constant. More generally, the asymptotics of the determinants of laplacians on rectilinear regions of $ \mathbb{Z}^2 $ have been computed in \cite{Kenyon2000}.  When appropriately rescaled, this graph laplacian can be understood as a finite difference approximation of the laplacian $ \frac{\partial^2}{\partial x^2}+ \frac{\partial^2}{\partial y^2} $ of Eulidean space. 

In this paper, we study a generalized graph laplacians in $ \mathbb{Z}^2 $ that can be regarded as finite difference approximations of Laplace-Beltrami operators induced by a Riemannian metric $ g $. These laplacians arise as follows: on any graph $ X $, the cochain complex $ ( C^*(X), d) $ can be regarded as a combinatorial version of the de Rham complex. We define a Riemannian structure on $ X $  as a set of inner products $ \langle , \rangle_{C^*(X)} $ on the cochain spaces. On such a graph with Riemannian structure, we define a combinatorial Guassian field theories, as outlined in section \S \ref{subsection:fieldt}. The partition function of the boson field theory is $Z = \det \Delta $, where the combinatorial laplacian $ \Delta $ is by $ \langle df, df \rangle_{C^*(X)} = \langle \Delta f ,f \rangle_{C^*(X)} $ for all $ f \in C^0(X) $. 

Assume now that the underlying space $ M \subset \mathbb{R}^2 $ of $ X $ is endowed with a metric $ g = g_{xx} \; dx^2 + g_{yy}\; dy^2 $. By square subdivision of $ X $, we have a sequence $ \{ X^i \} $ of graphs with mesh $ \epsilon_i \rightarrow 0 $. Assume that the $ \{ X^i \} $ have Riemmanian structure such that the combinatorial laplacian $ \Delta^i $ on $ X^i $ is a good approximation of the laplacian induced by $ g$ on $ M$ (see section \ref{sec:lapcons}). This can be regarded as discretization of a massless scalar field theory on a nontrivial background metric. We show that the determinant of $ \Delta^i $ has the asymptotic expansion (Theorem \ref{thm:maintheorem}):
\begin{align} \label{eq:eq3}
\log \det \Delta^i = \frac{1}{\epsilon_i^2 } \int_M F(g) \; dx \; dy + \frac{1}{\epsilon_i}\int_{\partial M} B(g) \cdot ds + \frac{1}{2} \log(\epsilon^i)  + O(1)
\end{align}
where the functions $ F$ and $ B $ are given in equation (\ref{eq:densities}).

The discrete Guassian field theories are known to be conformally invariant in the scaling limit $ i \rightarrow \infty $, and consequently the asymptotic expansion of the form (\ref{eq:eq3}) is generally anticipated by the theory of finite size scaling \cite{cardy}. In particular, the coefficient of the logarithmic term is universally related to the universal charge of the limiting conformal field theory and the Euler characteristic of the $M$.

Some work has been done on relating the constant terms in the asymptotics of the determinant of discretized or finite difference laplacians with regularized functional determinants of analytic laplacians. In one-dimension, there are many positive results, for example \cite{Forman}, \cite{bfk}. In two dimensions, for some particular cases such as flat rectangles and flat torii, with particular choice of discretized laplacians, the constant term is known to be equal to the zeta-regularized determinant of the analytic laplacian, \cite{duplant},\cite{cjk}. We do not attempt to calculate the constant term in the asymptotic determinant (\ref{eq:eq3}) and its investigation for further work.

The structure of the paper is as follows. In \S\ref{sec:sec2}, we describe the notion of a Riemannian structure on a  graph, and define the combinatorial codifferential and combinatorial laplacian. We explain how a Riemannian manifold can be approximated by a sequence of graphs with Riemannian structure. Then we outline the definition of a combinatorial Gaussian free field theory on a graph and describe the relation to the dimer model of statistical physics. 
In \S \ref{sec:sec3}, we focus on the bosonic field theory on $ (M,g) $, when $ M $ is a rectangular region in $ \mathbb{R}^2 $ endowed with a diagonal metric $ g $, and set up the calculation of the determinant. Sections \S  \ref{sec:sec4} and \S  \ref{sec:sec5} are more technical sections: \S  \ref{sec:sec4} is devoted to estimating the inverse $ (\Delta^i)^{-1} $  as $ i \rightarrow \infty $; in \S  \ref{sec:sec5} we prove the main theorem, equation (\ref{eq:eq3}). Finally in \S  \ref{sec:sec6}, we conclude with some general remarks and further research directions.

\subsection{Acknowledgements} I would like to thank my adviser Nicolai Reshetikhin for suggesting to study asymptotics of determinants of laplacians and for his reading of several drafts. I also thank Niccolo S. Poulsen for many valuable discussions regarding combinatorial Laplacians and dimer models, and for proofreading some versions of this text.

\section{Preliminaries} \label{sec:sec2}
Let $ X $ be a cell decomposition of a simply connected region $ M \subset \mathbb{R}^2 $, with $ q$-cells $ X_q $ for $ q = 0,1,2 $. The space of $q$-chains $ C_q(X) $ for $ q = 0,1,2 $ is the real vector spaces with canonical basis $ X_q $. The dual space space $ C^q(X) = C_q(X)^* $ is the space of functions $ X_q \rightarrow \mathbb{R} $. There is an isomorphism $ C_q(X) \simeq C^q(X) $ induced by the canonical basis on $ C_q(X) $. For any subcomplex $ L \subset X $, there is a natural projection $ p: C^q(X) \rightarrow C^q(L) $ given by restriction to $ L $. Define $ C^q(X, L ) =  \text{ker}(p) \subset C^q(X) $. The boundary of $ X $  is the 1-dimensional subcomplex $ \partial X $ contained in $ \partial M $.

On an oriented complex, there is a natural boundary operator $\partial_q: C_q(X) \rightarrow C_{q-1}(X) $. The coboundary $ d_q: C^q(X) \rightarrow C^{q+1}(X) $ is the adjoint of $ \partial_q $ with respect to the pairing $ C^q(X) \times C_q(X) \rightarrow \mathbb{R} $. The coboundary operator satisfies $ d_{q+1} d_q = 0 $ and the complex $ (C^*(X), d) $ can be regarded as a discrete analog of the de Rham complex of a smooth manifold.

In the remainder, we assume for simplicity that $ X $ is embedded such that image of each $2$-cell is a simple convex polygon. In this case the $1$-skeleton of $ X $ then defines a simple planar graph embedded in $ M $.

\subsubsection{The dual and double of a graph}
We define the dual complex $\widetilde{X} $ as the complex found by truncating the standard dual, so that $\widetilde{X} $ is also a cell complex of $ M $ (see figure \ref{fig:fig1}). Note that as usual, there is a bijection from $ X_q $ and to $ ( \widetilde{X} \backslash \partial \widetilde{X} )_{2-q}$, but in addition the truncation gives a bijection of $ \partial X_q $ and $ \partial \widetilde X_{1-q} $. 

The double $ D(X) $ of $ X $ is, roughly speaking, the union of the $ X $ and $ * X $  (see figure \ref{fig:fig1}). Precisely, $ D(X) $ is the complex consisting of: a $0$-cell $ p_\sigma $ at the center of each $ \sigma \in X $; a $1$-cell $ [p_\sigma, p_\tau] $ for each pair $ \sigma \in X_i, \; \tau \in X_{i+1} $ with $ \sigma $ adjacent to $ \tau $; and a $2$-cell $ [ p_\alpha, p_{\beta_1} , p_\gamma, p_{\beta_2} ] $, for each set of $ \alpha \in X_0 $, $ \beta_1, \beta_2 \in X_1 $, $\gamma \in X_2 $ pairwise adjacent. Note that the graph of $ D(X) $  has a natural bipartite structure given by painting $ p_\sigma $ white if $ \sigma $ is a $1$-cell, and black otherwise. 
\begin{figure}[h]
\label{fig:fig1}
\begin{tikzpicture}

\coordinate (p0) at (-1.2,0);
\coordinate (p1) at (.1,-.2);

\coordinate (p3) at (-.2,1.2);
\coordinate (p4) at (-1,1);

\coordinate (p2) at (1,.5);

\coordinate (c0) at ($ .5*(p0)+.5*(p4) $);
\coordinate (c1) at ($ .5*(p0)+.5*(p1) $);
\coordinate (c2) at ($ .5*(p4)+.5*(p3) $);
\coordinate (c3) at ($ .5*(p1)+.5*(p3) $);
\coordinate (c4) at ($ .5*(p2)+.5*(p3) $);
\coordinate (c5) at ($ .5*(p1)+.5*(p2) $);

\coordinate (d0) at ($ .25*(p0)+.25*(p4)+.25*(p1)+.25*(p3) $);
\coordinate (d1) at ($ .3333*(p1)+.3333*(p3) + .3333*(p2) $);

\draw[color = black]
(p0)-- (p1)
(p1)--(p2)
(p1)--(p3)
(p2)--(p3)
(p0)--(p4)
(p4)--(p3);

\foreach \i in {0,...,4}
{
\fill (p\i) circle (2pt);
}

\end{tikzpicture} \; \; \;
\begin{tikzpicture}

\draw[color = lightgray]
(p0)-- (p1)
(p1)--(p2)
(p1)--(p3)
(p2)--(p3)
(p0)--(p4)
(p4)--(p3);

\foreach \i in {0,...,4}
{
\fill [color= lightgray] (p\i) circle (2pt);
}

\draw[color = black]
(p0)--(p1)
(p1)--(p2)
(p2)--(p3)
(p0)--(p4)
(p4)--(p3);

\fill (d0) circle (2pt);
\fill (d1) circle (2pt);
\fill (c0) circle (2pt);
\fill (c1) circle (2pt);
\fill (c2) circle (2pt);
\fill (c4) circle (2pt);
\fill (c5) circle (2pt);

\draw[color=black]
(d0)--(d1)
(d0)--(c0)
(d0)--(c1)
(d0)--(c2)
(d1)--(c4)
(d1)--(c5);
\end{tikzpicture} \; \; \;
\begin{tikzpicture}

\draw[color = black]
(p0)-- (p1)
(p1)--(p2)
(p1)--(p3)
(p2)--(p3)
(p0)--(p4)
(p4)--(p3);

\draw[color = black]
(d0)--(c0)
(d0)--(c1)
(d0)--(c2)
(d0)--(c3)
(d1)--(c3)
(d1)--(c4)
(d1)--(c5);

\foreach \i in {0,...,4}
{
\fill (p\i) circle (2pt);
}

\foreach \i in {0,...,5}
{
\fill[color=white] (c\i) circle (2pt);
}
\foreach \i in {0,...,5}
{
\draw[color=black] (c\i) circle (2pt);
}

\fill (d0) circle (2pt);
\fill (d1) circle (2pt);

\end{tikzpicture}
\caption{A graph, its dual, and its double.}
\end{figure}
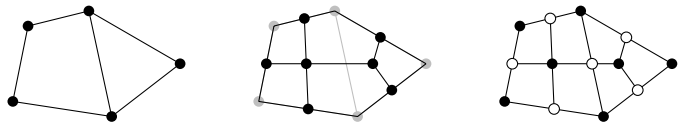
\subsection{Riemannian Structures on Graphs} Recall that the metric on a compact, oriented Riemannian manifold $ (M,g) $ induces a scalar product on the de Rham complex $ \Omega^*(M) $: for $ f,g \in \Omega^q(M) $
\begin{align*}
\langle f, g \rangle_{\Omega^q(M)} = \int_M f \wedge * g  
\end{align*}
 where $ *: \wedge^q T_p M \rightarrow \wedge^{2-q} T_p M $ is the Hodge dual. Continuing to view the cochain spaces $ C^q(X) $ as the combinatorial analog of $ \Omega^q(M) $, we define:
\begin{definition}
A Riemannian structure on $X$ is a set of scalar products $ \langle , \rangle_{C^q(X)} $  for $ q = 0,1,2 $.
\end{definition}
The scalar products can be required to satisfy certain additional axioms. In particular, recall that the scalar product on $ \Omega^q(M) $ is local in the sense that it is the integral of a fiberwise scalar product. We consider more generally the combinatorial notion of a fiberwise map $ \Omega^q(M) \rightarrow \Omega^r(M) $ as follows:
\begin{definition} \label{def:loc2} 
A map $ T: C^q(X) \rightarrow C^r(X) $ is said to $ m $-local if the matrix elements $ T_{\sigma, \tau} $ in the canonical basis of $ C^*(X) $ vanishes whenever the graph theoretic distance between $ p_\sigma, p_\tau \in D(X) $ is greater than $ 2 m $.
\end{definition}

\begin{definition}
A Riemannian structue $ \langle , \rangle_{C^q(X)}$ is local if the corresponding dual operator $ C^q(X) \rightarrow C^q(X) $  is $m$-local for some $ m $. A Riemannian structure is said to be diagonal if it is $0$-local.
\end{definition}
A diagonal Riemannian structure can be understood as an assignment of a volume weight $ | \sigma | = \langle \sigma, \sigma \rangle_{C^q(X)} $ to each cell $ \sigma \in X_q $. 

\subsubsection{Approximation theory} \label{sec:approx}
We now consider how a Riemannian manifold can be approximated by a graph with Riemannian structure. We assume as before that $ X $ is a cell decomposition of a region $ M \subset \mathbb{R}^2 $, but in addition assume that $ M $  is endowed with a smooth Riemannian metric $ g $. Define the mesh of $ X $  is $ \epsilon_X = \sup \{\; \text{diam}(\tau) \; | \; \tau \in X \} $. The embedding of $ X $ in $ M $ defines a surjection, namely the de Rham integration map $ I_{X}: \Omega^q(M) \rightarrow C^q(X) $:
\begin{align*}
I_X \omega = \sum_{\alpha \in X_q} \alpha \int_\alpha \omega 
\end{align*}
for $ \omega \in \Omega^q(M) $. 

Let $ \{X^i\}_{i \in \mathbb{N}} $ be a sequence of complexes in $ M $, such that $ \epsilon_{X^i} \rightarrow 0 $, for example generated by iterated subdivision of an initial complex. A sequence of Riemannian structures on the $ \{X^i\} $ is said to be consistent if the integration maps $ I_{X^i} $ approaches an isometry of spaces $ \Omega^q(M) $ and $ C^q(X^i) $ as $ i \rightarrow \infty $. More precisely:
\begin{definition}
A sequence of Riemannian structures on $ \{ X^i \} $ is consistent if for any $ f \in \Omega^q(M) $ there exists a constant $ \eta_f $ such that, for all $ i $:
\begin{align*}
\Big| \langle I_{X^i} f, I_{X^i} f \rangle_{C^q(X^i)} - \langle f, f \rangle_{\Omega^q(M)} \Big| \leq \eta_f \; \epsilon_{X^i}
\end{align*}
\end{definition}

\subsection{The Combinatorial Laplacian}
Let $ X $ be a complex with a Riemannian structure. Let $ L \subset \partial X $ be a subcomplex of $ \partial X $. Let $ d_{q,L} : C^q(X,L) \rightarrow C^q(X,L)   = \pi_L d $ where $ \pi_L : C^q(X) \rightarrow C^q(X,L) $ is orthogonal projection. The combinatorial codifferential combinatorial codifferential $ \delta_{q,L}: C^q(X) \rightarrow C^{q-1}(X) $ is defined the adjoint of $ d_q $:
\begin{align*}
\langle d_{q,L} f , g \rangle_{C^{q+1}(X,L)} = \langle f, \delta_{q+1,L} g \rangle_{C^q(X,L)} 
\end{align*}
for all $ f \in C^q(X), g \in C^{q+1}(X) $.
\begin{remark}
Choosing the canonical basis for the $ C^*(X) $ and writing the scalar products $ \langle , \rangle_{C^q(K)} $ as matrices $ g_q $, the codifferential $ \delta_q $ can be written in matrix form as:
\begin{align*}
\delta_q = g_{q-1}^{-1} d_q^{T} g_{q} 
\end{align*}
From this, it is clear that generally $ \delta_q $ is $ m $-local (for any $ m $!) only if the scalar product $ g_{q-1} $ is diagonal.
\end{remark}

\subsubsection{Duality} In this section, we briefly explain how Hodge duality arises in this combinatorial framework. Let $ X $ be a complex with a diagonal Riemannian structure, and as above $ L \subset \partial X $. Let $ \widetilde{L} \subset \partial \widetilde{X} $ be the complement of the image of $ L$ under the bijection $ \partial X_q \rightarrow \partial \widetilde{X}_{1-q} $, and $ \widetilde{X}_L $ be the complex found by removing from $ \widetilde{X} $ the image of $ L$ under both bijections $ X_q \rightarrow (\widetilde{X} \backslash\partial\widetilde{X})_{2-q} $ and $ \partial X_q \rightarrow \partial \widetilde{X}_{1-q} $.

There is a bijection of $ X_q \backslash L_q $ and $ X_{2-q} \backslash \widetilde{L}_{2-q} $, which induces an isomorphism $ \phi:  C_q(X,L) \rightarrow C_{2-q}(\widetilde{X}_L, \widetilde{L} )$. Putting a Riemannian structure on $ \widetilde{X}_L $ via the isomorphism, $ \phi $ becomes an isomorphism of chain complexes $ (C_q(X,L), d_L ) $ and  $ ( C_{2-q}  (\widetilde{X}_L,\widetilde{L}), \delta_{\widetilde{L}}) $; ie. $ \phi \circ d_{q,L} = \widetilde{\delta}_{q,\widetilde{L}} \circ \phi $ and $ \phi \circ \delta_{q,L} = \widetilde{d}_{q,\widetilde{L}} \circ \phi $. 

\subsubsection{The Combinatorial Laplacian} The combinatorial Laplacian is defined as $ \Delta_{q,L}= \delta_{q+1,L} d_{q,L} + d_{L,q-1} \delta_{q,L} $. This Laplacian should be understood as corresponding to mixed boundary conditions, with relative (Dirichlet) boundary conditions on $ L$ and absolute (Neumann) boundary conditions on $ \partial X \backslash L $. In particular, since we focus attention on $ 0$-forms, we let:
\begin{align*}
\Delta_L = \delta_{1,L} d_{0,L}
\end{align*}

For a diagonal Riemannian structure, it follows by straightforward computation that for $ p \in X_0 \backslash L $, the combinatorial laplacian can be written:
\begin{align}\label{eq:comblap}
\Delta_L f(p) = \sum_{ \substack{ \sigma \in X_1 \\ p \; \cap \; \sigma \neq 0 \\ \sigma \; \cap \;  L = 0 } } \frac{ \left| \sigma \right|}{|p|} \left( f(q) - f(p) \right)
\end{align}

\subsubsection{Consistency of the Combinatorial Laplacian} \label{sec:lapcons}
 Fix a submanifold $ L \subset \partial M $, and let $ \Delta $ be the Laplace-Beltrami operator acting on $C^\infty $ functions with Dirichlet boundary conditions on $L $ and Neumann on $ \partial M \backslash L $. Consider a sequence of graphs $ \{ X_i \} $ approximating $ (M,g) $ as in Section $ \ref{sec:approx} $, and such that for each $ i $ there is a subcomplex $ L_i \subset \partial X_i $ that is a decomposition of $ L $. In this context:

\begin{definition} The sequence of combinatorial laplacians $ \Delta_i $  on $ X_i $ is said to be consistent if for any $ f \in C^\infty(M) $ there exists a constant $ \eta_f $ such that, for all $ i $:
\begin{align} \label{eq:clap}
\Big| \Delta_L^i I_{X_i} f  - I_{X_i} \Delta f \big | \leq \eta_f \; \epsilon_{X_i}^2
\end{align}
\end{definition}

\subsection{Gaussian Field Theories on Graphs} \label{subsection:fieldt}
In this section, we briefly outline the definition of discrete Guassian field theories on graphs with Riemannian struture.
\subsubsection{Bosonic Free Field Theory}
Let $ X $ be a complex with a Riemannian structure. The (massless) free field action, or Dirichlet energy, is the quadratic form $S: C^0(X) \rightarrow \mathbb{R}^+ $:
\begin{align*}
S(f) = \langle df, df \rangle_{C^1(X)} 
\end{align*}
for $ f \in C^0(X) $.

The space of fields $ \mathcal{H} \subset C^0(X) $ of the discrete field theory is given by prescribing boundary conditions as follows: let $ L \subset \partial X $ be a nonempty subcomplex  and let $ \mathcal{H} = C^0(X, L) $.The Bosonic scalar field theory is then defined by the partition function:
\begin{align*}
Z_b(X,L) = \int_\mathcal{H}  e^{-S(\psi)} \; \mathcal{D}\psi
\end{align*}
where the measure $ \mathcal{D} \psi $ on $ \mathcal{H} $ is induced by the scalar product on $ C^0(X) $. The correlation functions are defined in the usual way as:
\begin{align*}
\langle \psi(p_1) \cdots \psi(p_n) \rangle = \frac{1}{Z_b(X,L)} \int_\mathcal{H} \psi(p_1) \cdots \psi(p_n)  e^{-S(\psi)} \; \mathcal{D} \psi
\end{align*}
Using standard Gaussian integral formulae give:
\begin{align*}
\log Z_b(X,L) = -\frac{1}{2} \log{\det(\Delta_L) } +\frac{1}{2} \left( \dim \mathcal{H} \right) \log( 2 \pi)
\end{align*}
Similarly, correlation functions can be expressed in terms of matrix elements of $ \Delta^{-1}_L $. For exampe, the two point function:
\begin{align*}
\langle \psi(p_1) \psi(p_2) \rangle = \left(\Delta_L^{-1} \right)_{p_1,p_2}
\end{align*}

\subsubsection{ Fermionic Free Field Theory}
As before, let $ L \subset \partial X $ be a subcomplex of $ \partial X $. Let $ S^+ = C^0(X,L) \oplus C^2(X,L) $ and $ S^- = C^1(X,L)$. The combinatorial Dirac operator $ D_L: S^+ \rightarrow S^- $ is the restriction of $ d_L + \delta_L $ to $ S^+ $, and the adjoint $ D_L^\dagger $  is the restriction of $ d_L + \delta_L $ to $ S^- $. This operator should be regarded as a combinatorial analog of the Hodge-Dirac operator.

The space of fields for the free fermionic field theory is $ \mathcal{H^+} = C^0(X, L) \oplus C^2(X,L) $ and $ \mathcal{H^-} = C^1(X,L) $.  The theory is defined by partition function given by the Grassmanian integral over $ \mathcal{H} = \wedge( \mathcal{H}^+ \oplus \mathcal{H}^- ) $:
\begin{align} \label{eq:ferm}
Z_f(X,L) = \int_{\mathcal{H}^+} \int_{\mathcal{H}^-} \exp\left( \frac{1}{2} \langle \psi, D_L \chi \rangle \right) \mathcal{D} \chi \mathcal{D} \psi
\end{align}
where the measure $ \mathcal{D} \chi \; \mathcal{D} \psi $ is induced by the scalar products on $ S^+ $ and $ S^- $. Correlation functions are once again defined by insertions into the path integral:
\begin{align*}
\langle \psi(p_1) \chi(p_1) \cdots \psi(p_n) \chi(p_n) \rangle& \\ = \frac{1}{Z_f(X,L) } \int_{\mathcal{H}^+} & \int_{\mathcal{H}^-}  \psi(p_1) \chi(p_1) \cdots \psi(p_n) \chi(p_n) \exp\left( \frac{1}{2} \langle \psi, D_L \chi \rangle \right) \mathcal{D} \chi
\end{align*}

We can explicitly calculate as follows.  Choosing the standard basis $ \{ \chi_i \} = X_1 \backslash \partial X_1$ for $ \mathcal{H}^- $ and $ \{ \psi_j \} = X_0 \backslash \partial X_0 \cup X_2  $ for $ \mathcal{H}^+ $, one can write the measure as:
\begin{align*}
\mathcal{D} \chi \mathcal{D} \psi = \prod_{i} \frac{1}{ \sqrt{ | \chi_i | } } \; d \chi_i \; \; \prod_{j}  \frac{1}{\sqrt{| \psi_j | } } \; d \psi_i
\end{align*}
Using the standard Grassman integral formulae, the partition function (\ref{eq:ferm}) can be evaluated as:
\begin{align} \label{eq:ferm}
Z_f(X,L) = \frac{ \prod  \sqrt{ | \psi_i | }  } {  \prod \sqrt{ | \chi_j | } }  \det \left( D_L \right)
\end{align}

The fermionic field theory is closely related to the bosonic theory. First define the dual fermionic theory as:
\begin{align}
\widetilde{Z}_f(X,L) = \int_{\mathcal{H}^+} \int_{\mathcal{H}^-} \exp\left( \frac{1}{2} \langle D_L \psi, \chi \rangle \right) \mathcal{D} \chi \mathcal{D} \psi
\end{align}
Then it is straightforward to show that:
\begin{align*}
Z_f(X,L) \widetilde{Z}_f(X,L)  = \big( Z_b(\widetilde{X}_L, \widetilde{L} ) \; Z_b( X , L)  \big)^{-1}
\end{align*}

\subsubsection{The Dimer Model}
The free fermionic theory defined above is closely related to the dimer model of statistical mechanics. First we briefly recall the definition of the dimer model. On a graph $  G $ with positive edge weights $ w \in C^1(G) $, a perfect matching $ M $ of $ G $ is a subset $ M \subset G_1 $ of edges such that each vertex $ v $ is contained in exactly one edge $ e \in M $. The weight of the matching $ M $ is  $ \prod_{e \in M } w(e) $. This defines a Gibbs measure $\mu $ on the set $ \mathcal{M} $ of all perfect matchings: the probability of a matching $ M $  is:
\begin{align*}
\mu(M,w) =  \frac{w(M)}{Z(G,w)}
\end{align*}
where the partition function $ Z(G,w) = \sum_{M \in \mathcal{M}} \omega(M) $. A gauge transformation of the dimer model is a function $ \lambda \in C^0(G) $ that acts on the edge weights as $ (\lambda \cdot w) \big( [p_1,p_2] \big) = w( [p_1,p_2]) \lambda(p_1) \lambda(p_2) $. The gauge transformation leaves the measure $ \mu $ invariant, but transforms the partition function as:
\begin{align} \label{eq:gauge}
Z \left( G, (\lambda \cdot w) \right)  = \left(\prod_{p \in G_0} \lambda(p) \right)Z(G,w)
\end{align}

The dimer model is exactly solvable by the Kasteleyn technique as follows: for simplicity, assume $ G $ is bipartite with black vertices $ \{ b_i \}_{i = 1 \cdots n} $ and white vertices $ \{ w_i \}_{j = 1 \cdots n} $. A Kasteleyn orientation is a function $ \widetilde{w}: G_1 \rightarrow \{ 1,-1 \} $, such that the product of $ \widetilde{w} $ on edges around any face is either $ -1 $ if the number of edges is divisible by four, or $ +1 $ otherwise. The Kasteleyn matrix is the $ n \times n $  matrix with columns and rows indexed by black and white vertices respectively, and with matrix elements
\begin{align*}
K_{b_i w_j} = \widetilde{w}( [ b_i w_j ] ) w( [ b_i, w_j ]) 
\end{align*}
when $ b_i $ and $ w_j $ are adjacent, and zero otherwise.

Kasteleyn's theorem is that the partition function of the dimer model is given by $ Z(G,w) = \det(K) $. In addition, the Kasteleyn matrix encodes the statistics of dimer model via the local statistics theorem, which gives correlation functions of the dimer model in terms of the inverse Kasteleyn matrix.

We now define a dimer model on the double of a complex $ X $ endowed with a local Riemannian structure. As before, let $ L \subset \partial X $ be a non-empty subgraph of the boundary. Define an the edge weight function $ w_r \in C^1(D(X))$ on the double as:
\begin{align*}
w_r( [p_\sigma, p_\tau] ) &= \sqrt{ \frac{\langle \sigma, \sigma \rangle_{C_{q}(G)}}{\langle \tau, \tau \rangle_{C_{q+1}(G)}} } 
\end{align*}
In writing the above formula, we have assumed $ \sigma $ is of smaller dimension than $ \tau $  in $ X $. Let $ D(X)' $ be the graph found by attaching to each $ p_\sigma $ with $ \sigma \in \partial X $ a unit weighted edge and an additional univalent vertex. (See figure \ref{fig:fig3}). Let $  Z_d(X, w_r,  L) $ be the dimer partition function of the dimer model defined by the graph $ D(X)' $ with weight $ w_r $. 

\begin{figure}[h]
\label{fig:fig3}
\begin{tikzpicture}

\coordinate (p0) at  (-1.5,.5);
\coordinate (p1) at (-1,-.2);
\coordinate (p2) at (-1,1.3) ;
\coordinate (p3) at (0,-.3);
\coordinate (p4) at (-.3,.5);
\coordinate (p5) at (.7,1.1);
\coordinate (p6) at (1,0);

\coordinate (c0) at ($ .5*(p0)+.5*(p1) $);
\coordinate (c1) at ($ .5*(p0)+.5*(p2) $);
\coordinate (c2) at ($ .5*(p2)+.5*(p4) $);
\coordinate (c3) at ($ .5*(p4)+.5*(p5) $);
\coordinate (c4) at ($ .5*(p5)+.5*(p6) $);
\coordinate (c5) at ($ .5*(p6)+.5*(p3) $);
\coordinate (c6) at ($ .5*(p3)+.5*(p1) $);
\coordinate (c7) at ($ .5*(p3)+.5*(p4) $);
\coordinate (c8) at ($ .5*(p2)+.5*(p5) $);

\coordinate (d0) at ( $ .2*(p0)+.2*(p1)+.2*(p3)+.2*(p4)+.2*(p2) $ );
\coordinate (d1) at ( $ .25*(p3)+.25*(p4)+.25*(p5)+.25*(p6) $ );
\coordinate (d2) at ( $ .3333*(p2)+.3333*(p4)+.3333*(p5) $ );

\draw[color=black]
(p0)--(p1)
(p0)--(p2)
(p2)--(p4)
(p1)--(p3)
(p3)--(p4)
(p4)--(p5)
(p3)--(p6)
(p5)--(p6)
(p2)--(p5);

\draw[very thick]
(p3)--(p6)
(p6)--(p5);

\foreach \i in {0,...,6}
{
\fill (p\i) circle (2pt);
}

\end{tikzpicture} \; \; \; \begin{tikzpicture}

\draw[color=black]
(p0)--(p1)
(p0)--(p2)
(p2)--(p4)
(p1)--(p3)
(p3)--(p4)
(p4)--(p5)
(p3)--(p6)
(p5)--(p6)
(p2)--(p5);

\draw[color=black]
(d0)--(c1)
(d0)--(c0)
(d0)--(c2)
(d0)--(c7)
(d0)--(c6)
(d1)--(c7)
(d1)--(c5)
(d1)--(c4)
(d1)--(c3)
(d2)--(c3)
(d2)--(c2)
(d2)--(c8);

\foreach \i in {0,...,6}
{
\fill (p\i) circle (2pt);
}

\foreach \i in {0,...,8}
{
\fill[color=white] (c\i) circle (	2pt);
\draw (c\i) circle (2 pt);
}

\foreach \i in {0,...,2}
{
\fill (d\i) circle (2pt);
}

\coordinate (b0) at (1.3,1.2);
\coordinate (b1) at (1.5,-.2);
\coordinate (b2) at (.1,-.8);
\coordinate (b3) at (1.4,.6);
\coordinate (b4) at (.7,-.7);

\foreach \i in {0,...,4}
{
\fill (b\i) circle (2pt);
}

\draw[color=black]
(p5)--(b0)
(p6)--(b1)
(p3)--(b2)
(c4)--(b3)
(c5)--(b4);

\end{tikzpicture}
\caption{A graph with boundary conditions in bold and the associated dimer model.}
\end{figure}
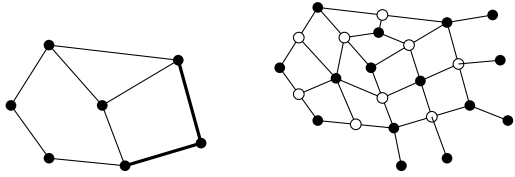

\begin{prop}
\begin{align*}
Z_d(X,w_r, L) = Z_f(X,L)
\end{align*}
\end{prop}

 This proposition states essentially that $ D_\mathcal{L}$ is a Kasteleyn matrix for the dimer model $ ( D(X)', w_r, L) $. It follows then that the correlation functions of the fermionic theory also give the correlation functions of the dimer model. 
\begin{proof} First, note that the dimer model on $ D(X)' $ is equivalent to the dimer model found by removing from $ D(X) $ vertices $ p_\sigma $ with $ \sigma \in L $, and edges adjacent to them. In addition, applying the gauge transformation
\begin{align*}
\lambda( p_\sigma) =\begin{cases} 1 / \sqrt{ | \sigma | } \ &\mbox{if } \sigma \in X_0, X_2 \\
 \sqrt{ | \sigma | }   &\mbox{if } \sigma \in X_1 \\ \end{cases}  
\end{align*}
on the remaining vertices transforms the edge weights into:
\begin{align*}
w' \left( [p_\sigma, p_\tau] \right) &= \begin{cases} 1 &\mbox{ if } \sigma \in X_0, \tau \in X_1 \\
 | \sigma |  / | \tau |    &\mbox{if } \sigma \in X_1, \tau \in X_2 \\ \end{cases} 
\end{align*}
Comparing the action of the gauge transformation (\ref{eq:gauge}) and the explicit formula for the fermionic partition function (\ref{eq:ferm}), it remains to show that:
\begin{align*}
Z_d(D(X) \backslash L, w') = \det( D_L )
\end{align*}
Recall that the double $ D(X) $ has a bipartite structure. Regarding $ D_\mathcal{H} $ as map operator from white vertices in  $ D(X)$  (edges in $ X $ ) to black vertices $ D(X)$ (vertices and faces in $ X$ we compute for $ \tau \in X_1 $:
\begin{align*}
D \tau = \sum_{\sigma \in X_0} (-1)^{(\sigma, \tau)} \sigma + \sum_{\sigma \in X_2} (-1)^{(\sigma, \tau)} \frac{| \sigma | }{ | \tau |}  \sigma
\end{align*}
where $ (-1)^{(\sigma, \tau)} $ is the function that is 0 if $ \sigma $ is not adjacent to $ \tau $, and is $ \pm 1 $ if the orientations of $ \sigma $ and $ \tau $ are compatible.
That is, up to a sign, the matrix element $ D_{\sigma,\tau} $ is equal to the edge weight $ w'([p_\sigma, p_\tau]) $. It remains to check that the signs of the matrix elements of $ D $ satisfy the Kasteleyn orientation condition. Since every face in the double $ D(X) $ is quadrilateral, the Kasteleyn condition requires the number of minus signs around each is odd. This follows from $ d^2 = 0 $ and $ \delta^2 = 0 $. 
\end{proof}

\section{Combinatorial Laplacians in Rectangular Graphs} \label{sec:sec3}
In the remainder of this paper, we focus on the particular case that $ U $ is a rectangular region with a diagonal metric $ g = g_{xx}\; dx^2 + g_{yy} \; dy^2 $, approximated by subgraphs of $ \mathbb{Z}^2 $. For simplicity, we assume that both the metric and the induced Riemannian distance function are analytic. In addition, denote by $\partial M_i $ for $ i = 1, \cdots, 4 $ be the four straight segments of $ \partial M $, and let boundary conditions be specified by $ L \subset \partial M $, where $ L$ is the union of a non-empty subset of $ \{ \partial M _i \} $.

Let $ X^0 $ be the cell decomposition of $ M $, consisting of four $0$-cells placed at the corners of $ U$, four $1$-cells, and one $ 2 $-cell. Let $ \{ ( X^i, \langle , \rangle_{C^q(X^0} ) \}_{i \in \mathbb{N} } $ be a sequence of subdivisions by iteratively taking doubles, endowed with a local, consistent Riemannian structure and consistent combinatorial laplacians $ \Delta^i_L $. 

\begin{thm} \label{thm:maintheorem}
In the above setting, the determinant of the combinatorial laplacian $ \Delta^i $ has the following expansion:
\begin{align} \label{eq:finalasym}
\log \det \Delta^i = \frac{1}{\epsilon_i^2 } \int_M F(g) \; dx \; dy + \frac{1}{\epsilon_i}\int_{\partial M} B(g) \cdot ds + \frac{1}{2} \log(\epsilon^i)  + O(1)
\end{align}
where:
\begin{align} \label{eq:densities}
F(g_{xx},g_{yy}) &= \frac{2}{\pi}  \Im \left( \text{\emph{Li}}_2 \left(i \sqrt{\frac{g_{yy}}{g_{xx}}} \right) + \text{\emph{Li}}_2 \left(i \sqrt{\frac{g_{xx}}{g_{yy}}}\right) \right) - \log\left(\sqrt{g_{xx} g_{yy}}\right) \notag \\
B_x(g_{xx},g_{yy}) &= -\frac{1}{2} \log( \sqrt{g_{xx} g_{yy}} + \sqrt{ g_{xx} g_{yy} + g_{yy}} ) \\
B_y(g_{xx},g_{yy}) &= -\frac{1}{2} \log( \sqrt{g_{xx} g_{yy}} + \sqrt{ g_{xx} g_{yy} + g_{yy}} ) \notag
\end{align}
where $ \text{\emph{Li}}_2 $ is the dilogarithm function.
\end{thm}

The approach to calculate the determinant is to study variations of the metric $ g $. Under a variation $ g \rightarrow g + \delta g $  we have:
\begin{align}\label{eq:varf}
\frac{\delta}{\delta g} \log \det \Delta^i_ L =  \text{Tr}\left( (\Delta^i_L)^{-1} \frac{ \delta \Delta_L^i}{\delta g} \right)
\end{align}
In the next section \S \ref{sec:sec4}, we compute an estimate of the combinatorial Green's function $ \Delta^i_L $, which we use in section \S \ref{sec:sec5} to compute an estimate of (\ref{eq:varf}). The asymptotic expansion (\ref{eq:finalasym}) follows by integrating (\ref{eq:varf}), with the constant of integration fixed by the expansion (\ref{eq:eq2}), corresponding to the case that $ g = dx^2 + dy^2 $.

We conclude this section by estimating the matrix elements of the combinatorial laplacians $ \Delta^i_L $ as $ i \rightarrow \infty $.

\subsubsection{Asymptotics of the Combinatorial Laplacian}
 First, we fix some useful notation borrowed from the finite difference calculus. For clarity, we leave implicit superscripts indicating the level of subdivision when there is no danger of confusion, for example $ C^0(X) $ for $ C^0(X^i) $, $ \e $ for $ \e_{X^i} $, etc.  For $ p \in X_0 \backslash \partial X_0 $ with coordinates $ (x,y) \in \mathbb{R}^2 $, define the finite difference operators:
\begin{align*}
D_x f(x,y) &= \frac{1}{2} \left( f(x+\e, y) - f(x-\e, y) \right) \\
D_x^2 f(x,y) &= f(x+ \e,y) - 2 f(x,y) + f(x-\e, y)
\end{align*}
and similar operators in the $ y$-directions. These operators are consistent:  for every $ f \in C^\infty(U,\partial U) $, there are constants $ \eta_{f} $ such that
\begin{align} \label{eq:cfindif}
\begin{split}
\left| \frac{1}{\e} D_x I_{X} f(x, y)  - I_{X} \partial_x f(x,y) \right| \leq \e^2 \eta_f  \\
\left| \frac{1}{\e^2} D^2_x I_{X} f(x,y)  - I_{X}\partial_x^2 f(x,y) \right| \leq \e^2 \eta_f 
\end{split}
\end{align}
\begin{prop} 
Then the laplacian has the following asymptotic asymptotic expansion:
\begin{equation}
\begin{aligned} \label{eq:lapasym}
\epsilon^2 \Delta^i = \frac{1}{\sqrt{ g_{yy} g_{xx} }} \left( \sqrt{\frac{g_{yy}}{g_{xx}}} \; D_x^2 + \sqrt{\frac{g_{xx}}{g_{yy}}} \; D_y^2  + \e \; \partial_x \sqrt{\frac{g_{yy}}{g_{xx}}} \;  D_x + \e \; \partial_y \sqrt{\frac{g_{xx}}{g_{yy}}} \; D_y \right)  \\ + O(\epsilon^2)  D_x^2  + O(\epsilon^2) D_y^2 + O(\epsilon^3) D_x + O(\epsilon^3) D_y 
\end{aligned}
\end{equation}
\end{prop}
\begin{proof}
It is straightforward to check that on the square lattice, the combinatorial laplacian (\ref{eq:comblap}) can be written as a sum of finite difference operators $D_x^2, D_y^2, D_x$ and $D_y$. (In fact, any operator that couples only nearest neighbor vertices and and annihilates constant functions can be written in terms of these finite difference operators). The proposition then follows from comparing the consistency of the finite difference operators (\ref{eq:cfindif}), and consistency of the laplacian (\ref{eq:clap}).
\end{proof}

\section{Inverse of Combinatorial Laplacian} \label{sec:sec4}
In this section, we derive an estimate for the combinatorial Green's function $ \Delta_L^i $. The approach can be summarized as follows. Since the combinatorial Laplacian is a consistent approximation of $ \Delta $, the smooth Green's function $ G $ is an approximation of the combinatorial Green's function where it is smooth. However, the Green's function $ G(p,q) $ is logarithmically divergent where $ p \rightarrow q $. We construct an approximate inverse matrix $ \widetilde{G}^i $ by replacing these logarithmic divergences with a lattice regularized logarithm function. We show that $ \widetilde{G}^i $ is a "parametrix," in the sense that for sufficiently large $ i $, the error
\begin{align*}
R = I - \Delta^i G^i
\end{align*}
tends to zero, and for sufficiently large $ i $, the inverse $ \Delta_i^{-1} $ can be written as
\begin{align} \label{eq:series}
(\Delta^i_L)^{-1} = G^i ( I + R+ R^2 + \cdots)
\end{align}

\noindent
We begin first by recalling some facts about the Greens function of the smooth setting.

\subsection{Smooth Green's Function} \label{subsec:smooth}
Recall that when the Riemmanian distance function $ d: M \times M \rightarrow \mathbb{R} $ is smooth away from the diagonal, $ G $ can be written neatly in terms of $ d $. Define $ \omega: M \times M \rightarrow \mathbb{R} $ as satisfying
\begin{enumerate}
\item $ \Delta \omega(p_1,p_2) = 0 $
\item $ \omega(p_1,p_2) = - \frac{\sqrt{g(p_1) }}{4 \pi} \log (d^2(p_1,p_2) ) $ for $ p_2 $ or $ p_1 \in \partial M $
\end{enumerate}
Then the Green's function can be written:
\begin{align}
G(p_1,p_2) =  - \frac{\sqrt{g(p_1) }}{4 \pi} \log \left(d^2(p_1,p_2) \right)  + \omega(p_1,p_2)
\end{align}
\noindent
In this decomposition, both terms are divergent: the logarithmic term contains divergences as $ p_1 \rightarrow p_2 $ while $ \omega $ contains divergences as $ p_1,p_2 \rightarrow \partial M $.

To understand the divergent behaviour as $ p_1 \rightarrow p_2 $, one can expand the distance function using the Hamilton-Jacobi equation. Letting $ (x,y) = p_2 - p_1 $, we calculate that as $ (x,y) \rightarrow 0 $:
\begin{equation} \label{eq:dist}
\begin{aligned}
d^2(p_1,p_2) =&  g_{xx} \; x^2 + g_{yy} \; y^2  \\ &+ \frac{1}{2} \partial_x  g_{xx} \; x^3 +   \frac{1}{2} \partial_y  g_{xx} \; x^2 y+  \frac{1}{2} \partial_x  g_{yy} \; x y^2+ \frac{1}{2} \partial_y  g_{yy} \; y^3\cdots
\end{aligned}
\end{equation}
Then for the logarithm of the distance function:
\begin{align}
\log\left(d^2(p_1,p_2) \right) =  \log \left( g_{xx} x^2 + g_{yy} y^2 \right) + \log \left( 1 + \frac{O(| p_1-p_2 |^3 )}{g_{xx} \; x^2 + g_{yy} \; y^2 } \right)
\end{align}
as $ (x,y) \rightarrow 0 $. The term on the right is bounded, so we have isolated the divergence into the first term.

Turning to $ \omega $, recall that although for fixed $ p_1 $, $ \omega(p_1, \cdot) $ is smooth, as $ p_1 \rightarrow \partial M $ the function and its partial derivatives diverge. One can understand the behaviour of $ \omega(p_1,\cdot) $ as $ p_1 \rightarrow \partial M $ using the method of images.(citation?). Define the reflected and the corner distance functions as: 
\begin{align}
  d_{R_i}(p_1,p_2) = & \inf \{ d(p_1,b) + d(b,p_2) \; | \; b \in \partial M_i \} \\
 d_{R_{ij}}(p_1,p_2) =  &\inf \{ d(p_1,b_1) + d(b_1,b_2) +  d(b_2,p_2) \; \\& \hspace{40pt} | \; b_1 \in \partial M_i, b_2 \in \partial M_j \text{ or }  b_1 \in \partial M_j, b_2 \in \partial M_i \}  \notag
\end{align}
The reflected distance function is well defined for all $ p_2 $ when $ d(p_1,\partial M_i) $ is sufficiently small, and similarly for the corner distance function. By definition, they satisfy $ d_{R_i}(p_1,p_2) = d(p_1,p_2) $ when $ p_1 \in \partial M_i  $, and $ d_{R_{ij}}(p_1,p_2) = d_{R_i}(p_1,p_2) $ when $ p_2 \in \partial M_j $. In addition, they satisfy the Hamilton-Jacobi equation, so that $ \Delta \log( d_{R_i}(p_1,p_2)) = 0 $ and $ \Delta \log( d_{R_{ij}}(p_1,p_2))  = 0 $.

Fix $ \delta $ such that reflected distance function is well defined when $ d_{R_i}(p_1,\cdot) $ is well defined for all $ i $ and $ p_1 $ with $ d(p_1, \partial M_i) $. Let $ B_i  = \{ p | d(p, \partial M_i) < \delta \} $ for all i, $ B = \cup_i B_i $.  For $ i \neq j $ let $ C_{ij} = B_i \cap B_j $. Redefine $ \omega $ by the decomposition:
\begin{enumerate}
\item For $ p \in C_{ij} $:
\begin{align*}
G(p_1,p_2) = &-\frac{1}{4 \pi} \log \left( d^2(p_1,p_2)  \right ) + \frac{1}{4 \pi} \log \left( d^2_{R_i}(p_1,p_2) \right)  \\ & +  \frac{1}{4 \pi} \log \left( d^2_{R_j}(p_1,p_2) \right) -  \frac{1}{4 \pi} \log \left( d^2_{R_{ij}}(p_1,p_2) \right) + \omega(p_1,p_2)
\end{align*}
\item For $ p \in B_i $:
\begin{align*}
G(p_1,p_2) = &-\frac{1}{4 \pi} \log \left( d^2(p_1,p_2)  \right ) + \frac{1}{4 \pi} \log \left( d^2_{R_i}(p_1,p_2) \right) + \omega(p_1,p_2) 
\end{align*}
\item For $ p \not\in B $:
\begin{align*} 
 G(p_1,p_2) = -\frac{1}{4 \pi} \log \left( d^2(p_1,p_2)  \right ) + \omega(p_1,p_2)
\end{align*}
\end{enumerate}
With this redefinition, it can be shown from a Schauder estimate that $ \omega(p_1,\cdot) $ is smooth and bounded for all $ p_1 \in M $, so that the divergences have been pushed into the reflected distance functions.

We can compute the expansion of the reflected distance functions as before, using the Hamilton-Jacobi equations. Write the coordinates $ p_j = (x_j,y_j)  $ for $ j = 1,2 $, and for simplicity assume that $ \partial M_i $ is aligned with the y-axis, and $ \partial M _j $ is aligned with the x-axis. Then:
\begin{align*}
d^2_{R_i}(p_1,p_2) &=  g_{xx}(p_1) \; (x_1 + x_2)^2  + g_{yy}(p_1) \; (y_2-y_1)^2 + O(x_1^2) + O( | p_1 - p_2|^2 ) \\ 
d^2_{R_{ij}}(p_1,p_2) &=  g_{xx}(p_1) \; (x_1 + x_2)^2  + g_{yy}(p_1) \;(y_1+y_2)^2 + O(x_1^2 + y_1^2) + O( |p_1-p_2|^2)
\end{align*}
 so that:
\begin{align}\label{eq:image}
\log \left( d^2_{R_i}(p_1,p_2) \right)  &=  \log \left(g_x^{-1} (x_1 + x_2)^2  + g_y^{-1} (y_2-y_1)^2 \right)  + \cdots \\
\log \left( d^2_{R_{ij}}(p_1,p_2) \right) &= \log \left(  g_x^{-1} (x_1 + x_2)^2  + g_y^{-1} (y_1+y_2)^2 \right) + \cdots
\end{align}
where the remaining terms are continuous as $ p_1,p_2 \rightarrow \partial M $. 

\subsection{Combinatorial Green's Function}
Now we construct the parametrix $ G^i $  for the combinatorial Green's function. It is based on lattice regularized logarithm function, which is a special case of a general formula given in \cite{Kenyon2002}, arising in the theory of discrete complex analysis. It is defined as follows:
\begin{definition}
For $ a,b $ positive real numbers, define $ \L_{(a,b)}: \mathbb{Z}^2 \rightarrow \mathbb{R} $ for $ (x,y) \in \mathbb{Z}^2 $ as:
\begin{align} \label{eq:logdef}
\L_{(a,b)}(x,y) &= -\frac{\sqrt{a b}}{8 \pi^2 i } \int_C \frac{\log(z)}{z} \left(\frac{z+e^{ i w}}{z-e^{ i w}}\right)^{|x|+|y|}  \left(\frac{z+e^{-i w}}{z-e^{-i w}}\right)^{|x|-|y|} dz
\end{align}
where $ w = \tan^{-1} \left( \sqrt{b/a} \right) $ and where $ C $ is a curve that  encircles the points $ \pm e^{  iw } $ and $ \pm e^{- i w} $, but not the origin.
\end{definition}	
This function satisfies the following properties, the proofs of which can be found in \cite{Kenyon2002},\cite{Bucking}:
\begin{enumerate}[\hspace{10pt}1.]
\item $ \L_{(a,b)}(0,0) = 0 $
\item $ \left(\frac{1}{a} D_x^2 + \frac{1}{b} D_y^2 \right) \L_{(a,b)}(x,y) = \delta_{0x} \delta_{0y} $
\item $ \L_{(a,b)} (x,y) $ has the asymptotic expansion for $ x, y \rightarrow \infty $:
\begin{align*}
\L_{(a,b)}(x,y) & = -\frac{\sqrt{a b}}{4 \pi } \log \left( \frac{ 16 \gamma^2}{ a+b } \right)  - \frac{\sqrt{a b} }{4 \pi} \log\left(a x^2  + b y^2 \right) \\ & - \frac{\sqrt{a b}}{24 \pi} \frac{a b \left( b x^4 - (a + b ) x^2 y^2 + a y^4 \right)}{(b\;x^2 + a\;y^2)^3} + O\left(\frac{1}{(x^2 + y^2)^2} \right) 
\end{align*}
where $ \gamma_l = \exp( \gamma_e) $ is the exponential of the Euler gamma $ \gamma_e $.
\item 
The finite difference derivatives of $ \L_{(a,b)} $ have the asymptotic expansion:
\begin{equation} \label{eq:logasymp}
\begin{aligned}
D_x \L_{(a,b)}(x,y) &= \frac{\sqrt{a b}}{2 \pi} \frac{a x}{a x^2 + b y^2  } + O \left( \frac{1}{(x^2 + y^2)^{3/2}} \right) \\
D_x^2 \L_{(a,b)}(x,y) &= \frac{\sqrt{a b}}{2 \pi} \frac{a b \; y^2 - a^2 x^2}{ (a x^2 + b y^2 )^2 } + O \left( \frac{1}{(x^2 + y^2)^{2}} \right) \\
D_y \L_{(a,b)}(x,y) &= \frac{\sqrt{ a b}}{2 \pi} \frac{b y }{a x^2 + b y^2 } + O \left( \frac{1}{(x^2 + y^2)^{3/2}} \right) \\
D_y^2 \L_{(a,b)}(x,y) &= \frac{\sqrt{ a b}}{2 \pi} \frac{a b \; x^2  - b^2 y^2 }{(a x^2 + b y^2)^2} + O \left( \frac{1}{(x^2 + y^2)^{2}} \right)
\end{aligned}
\end{equation}
(Here it is understood that in the finite difference operators $ \epsilon = 1 $.)
\item At the origin, the derivative is:
\begin{align*}
D_x \L_{(a,b)}(0,0) &= 0 \\
D_x^2  \L_{(a,b)}(0,0) &= \frac{2 a}{  \pi} \tan^{-1} \left( \sqrt{\frac{b}{a}} \right) \\
D_y  \L_{(a,b)}(0,0)  &= 0 \\
D_y^2  \L_{(a,b)}(0,0) &=   \frac{2 b}{ \pi} \tan^{-1} \left( \sqrt{\frac{a}{b}} \right)
\end{align*}
\end{enumerate}

\subsubsection{Green's function on $ X^i $ }
We now utilize the lattice logarithm defined above to construct the Green's function on $ X^i $. Roughly speaking, we replace all logarithmic divergences identified in \S \ref{subsec:smooth} with corresponding lattice logarithm functions. For some notational simplicity, we fix:
\begin{align}
\begin{split}
L(p_1,p_2) &= -\frac{\sqrt{g(p_1) }}{4 \pi} \log \left( g_{xx}(p_1) x^2 + g_{yy}(p_2) y^2 \right) \\
\L(p_1,p_2) &= \L_{g(p_1)}\left( \frac{x}{\epsilon}, \frac{y}{\epsilon} \right)
\end{split}
\end{align}
where $ (x,y) = p_2 - p_1 $. We define also $ L_{R_i}, L_{R_{ij}}$ and $ \L_{R_i}, \L_{R_{ij}} $ in the obvious way based on the expansions (\ref{eq:image}).

Then, define $ G^i: X^i_0 \times X^i_0 \rightarrow \mathbb{R} $ as:
\begin{enumerate}
\item For $ p \in C_{ij} $:
\begin{align*}
G^i(p_1,p_2) =& \; G(p_1,p_2) - L(p_1,p_2)  + \L (p_1,p_2) 
\\ & - L_{R_{i}} +\L_{R_{i}}(p_1,p_2) 
\\ & - L_{R_{j}} +\L_{R_{j}}(p_1,p_2) 
\\ & + L_{R_{ij}} -\L_{R_{ij}}(p_1,p_2) 
\end{align*}
\item For $ p \in B_i $:
\begin{align*}
G^i(p_1,p_2) =&\; G(p_1,p_2) - L(p_1,p_2)  + \L (p_1,p_2) 
\\ & - L_{R_{i}} +\L_{R_{i}}(p_1,p_2) 
\end{align*}
\item For $ p \not\in B $:
\begin{align*}
G^i(p_1,p_2) =& \; G(p_1,p_2) - L(p_1,p_2)  + \L (p_1,p_2) 
\end{align*}
\end{enumerate}
Here, when $ p_1 = p_2 $ we set $ I_{X^i}( G(p_1,p_1) + L(p_1,p_1) ) = 0 $.

We now estimate the error $ R^i = \Delta^i G^i - I $. The series converges if in the operator norm $ | R^i |_{op} < 1 $.  We show the stronger condition that each column of $ R^i $ has $ L_1 $ norm of $ O(\epsilon) $. It follows then that for sufficiently large $ i $, $ | R^i |_{op} < 1 $. 
\begin{prop}
Let $ R: X^i_0 \times X^i_0 \rightarrow \mathbb{R} $ be defined as:
\begin{align*}
R^i(p_1,p_2) = \Delta^i G^i(p_1,p_2) - \delta_{p \cdot}
\end{align*}
Then there is a constant $ C $ such that:
\begin{align*}
\sum_{p_2 \in X^i_0}  | R^i(p_1,p_2) | \leq C h 
\end{align*}
\end{prop}

\begin{proof}
We assume that $ p_1 \not \in B $, and we split the sum
\begin{align}
\sum_{p_2 \in X^i_0}  | R^i(p_1,p_2)  | = R^i(p_1,p_1) +  \sum_{\substack{p_2 \in X^i_0 \\ p_2 \neq p_1}} | R^i(p_1,p_2) | \label{eq:p1t0}
\end{align}
When $ p_1 \neq p_2 $, we write $ R^i(p_1,p_2) $ using the fact that $ \Delta G(p_1,p_2)  = 0 $ as follows:
\begin{align}
| R^i  | &=  | \Delta^i G- \Delta^i L  + \Delta^i \L | \notag
\\ &=  | \Delta^i G - \Delta^i L- \e^2 (\Delta G - \Delta L)  + \Delta^i \L - \e^2 \Delta L| \notag
\\ &\leq  \big| \Delta^i G - \Delta^i L- \e^2 (\Delta G - \Delta L) \big|  + \big| \Delta^i \L - \e^2 \Delta L \big| \label{eq:p1t1}
\end{align}
We estimate the two terms seperately using the aymptotic expansion of ( \ref{eq:lapasym}) of $ \Delta^i$, along with the consistency properties of the finite difference operators for the first term, and the asymptotic formulas  (\ref{eq:logasymp}) of $ \L $  for the second term.

We begin with the first term: $ | \Delta^i G - \Delta^i L- \e^2 (\Delta G - \Delta L)| $. Recall that the errors of the first and second order finite difference operators are bounded by the third and fourth partial derivatives of $ G(p_1,p_2) $ and $ L(p_1,p_2) $, which diverge near the diagonal. However, we can estimate these divergent terms using the expansion of the distance function (\ref{eq:dist}). Letting $ (x,y) = p_2 - p_1 $, as $ (x,y)\rightarrow 0 $ we have:
\begin{equation} \label{eq:gdist}
\begin{aligned} 
G(p_1,p_2) - & L(p_1,p_2) \\ &= -\frac{\sqrt{g(p_1)}}{4 \pi} \log\left( 1 + \frac{ \frac{1}{2} \partial_x  g_{xx} \; x^3 +   \frac{1}{2} \partial_y  g_{xx} \; x^2 y+  \frac{1}{2} \partial_x  g_{yy} \; x y^2+ \frac{1}{2} \partial_y  g_{yy} \; y^3\cdots }{g_{xx} x^2 + g_{yy} y^2} \right)
\end{aligned}
\end{equation}
(we ignore $ \omega(p_1,p_2) $ since it is smooth). From this, we can calculate explicitly that the third and fourth partial derivatives behave as $ p_2 \rightarrow p_1 $ as:
\begin{align*}
\partial_x^m \partial_y^{3-m} G(p_1,p_2) -  \partial_x^m \partial_y^{3-m} L(p_1,p_2) =  O \left(\frac{1}{|p_2 - p_1|^2} \right)\\
\partial_x^m \partial_y^{4-m} G(p_1,p_2) -  \partial_x^m \partial_y^{4-m} L(p_1,p_2) =  O \left(\frac{1}{|p_2 - p_1|^3} \right)
\end{align*}
Here the constant of proportionalities in the asymptotics can be written in terms of the metric $ g $ and its partial derivatives, which we omit writing for brevity. Then we have:
\begin{align*}
\big| D_x^2 (G - L)  -  \e^2 \partial_x^2 (G-L) \big| = \frac{O(\e^4)}{ |p_2- p_1|^3 } \\
\big| \e D_x (G - L)  -  \e^2 \partial_x (G-L) \big| = \frac{O(\e^4)}{ |p_2- p_1|^2 }
\end{align*}
and similar expressions for the $ y $-direction. It follows from this and equation (\ref{eq:lapasym}) that:
\begin{align} \label{eq:p1t2}
| \Delta^i G (p_1,p_2)- \Delta^i L(p_1,p_2)- \e^2 (\Delta G(p_1,p_2) - \Delta L(p_1,p_2))| = \frac{O(\e^4)}{|p_2 - p_1|^3}
\end{align}

We turn now to $ | \Delta^i \L - \e^2 \Delta L|$. First observe that the equations (\ref{eq:logasymp}) can be rewritten as:
\begin{align*}
D_x^2 \L(p_1,p_2) = \e^2 \partial_x^2 L(p_1,p_2) + \frac{O(\e^4)}{x^4} \\
\e D_x\L(p_1,p_2) = \e^3 \partial_x L(p_1,p_2) + \frac{O(\e^4)}{x^3}
\end{align*}
and similar expressions for the $ y $-direction.  In addition, recall that for $ p_1 \neq p_2 $, we have  $ \left( \frac{D_x^2}{g_{xx}(p_1)} + \frac{D_y^2}{g_{yy}(p_1)} D_y^2 \right) \L(p_1,p_2) = 0 $ and $ \left( \frac{\partial_x^2}{g_{xx}(p_1)}  + \frac{\partial_y^2}{g_{yy}(p_1)} \right) L(p_1,p_2) = 0 $ Thus:
\begin{align*}
| \Delta^i \L - \e^2 \Delta L | = \big|\left( \Delta^i - \frac{D_x^2}{g_{xx}(p_1)} - \frac{D_y^2}{g_{yy}(p_1)} \right) \L - e^2 \left( \Delta -  \frac{\partial_x^2}{g_{xx}(p_1)}  - \frac{\partial_y^2}{g_{yy}(p_1)} \right) L \big| 
\end{align*}
Using again the expression  (\ref{eq:lapasym}) for $ \Delta^i $, it follows that:
\begin{align} \label{eq:p1t3}
 | \Delta^i \L(p_1,p_2) - \e^2 \Delta L(p_1,p_2) | = \frac{O(\e^4)}{ |p_2 - p_1 |^3}
\end{align}
Putting together (\ref{eq:p1t3}) and (\ref{eq:p1t2}), we find for $ p_1 \neq p_2 $:
\begin{align}
|R^i(p_1,p_2) | = \frac{O(\e^4)}{ |p_2 - p_1|^3} 
\end{align}
We can estimate the sum by changing variables $ p = (p_2 - p_1)  / \e $. Then:
\begin{align}
\sum_{\substack{p_2 \in X^i_0 \\ p_2 \neq p_1}} | R^i(p_1,p_2) | \leq O(\e) \sum_{p \in \mathbb{Z}^2 \backslash 0} \frac{1}{| p |^3} &\leq O(\e)\left( \int_1^\infty \frac{2 \pi |p| }{|p|^3} d |p| + 1\right) \notag \\ 
&= O(\e) \label{eq:p1t4}
\end{align}

Lastly, we have the case $ p_1 = p_2 $. We compute explicitly:
\begin{equation} \label{eq:p1t5}
\begin{aligned} 
\Delta^i \L(p_1,p_1) &= 1 + O(\e^2) \\
\Delta^i G(p_1,p_1)  - \Delta^i(p_1,p_1) &= O(\e^2)
\end{aligned}
\end{equation}

Putting together equations (\ref{eq:p1t3}) and (\ref{eq:p1t4}) completes the calculation. The case where $ p \in B $ is identical except for the additional "image" charges that can be handled in precisely the same way. 

\end{proof}

\section{Asymptotic Determinant of Laplacian} \label{sec:sec5}
In this section, we derive an expression (\ref{eq:vareq}) for the variation of the determinant of the combinatorial laplacian. It is more convenient to write variations of $ g $ in terms of scaling variations and . Letting $ \lambda = \sqrt{ g_{xx} g_{yy} } $ and $ \omega = \sqrt{ g_{yy} / g_{xx}} $, we rewrite (\ref{eq:lapasym}) as:
\begin{align*}
\Delta_L^i = \frac{1}{\lambda} \left( \omega \; D_x^2 + \omega^{-1} \; D_y^2  + \e \partial_x \omega \; D_x + \e \partial_y \omega^{-1} \; D_y \right)  \\ +  \; (O(\e^2) D_x^2 + O(\e^2) D_y^2 + O(\e^3) D_x + O(\e^3) D_y)
\end{align*}
Then we rewrite (\ref{eq:varf}) as:
\begin{align} \label{eq:trace}
 \text{Tr}\left( (\Delta^i_L)^{-1} \delta \Delta_L^i\right) = & \sum_{p \in X \backslash \partial X} -\frac{\delta\lambda}{\lambda} 
\Delta_L^i (\Delta^i_L)^{-1}(p,p) \notag \\ + & \frac{1}{\lambda} \sum_{p \in X \backslash \partial X}  \left( \delta \omega \; D_x^2 - \frac{ \delta \omega}{\omega^2} \; D_y^2 + \e \partial_x \delta \omega \; D_x -  \e \frac{\partial_y \delta \omega}{\omega^2} D_y \right) (\Delta^i_L)^{-1}(p,p) \\
+ & \sum_{p \in X \backslash \partial X}  \left(O(\epsilon^2) D_x^2 + O(\epsilon^2) D_y^2 + O(\e^3) D_x + O(\e^3) D_y \right) (\Delta^i_L)^{-1}(p,p) \notag
\end{align}
The finite difference derivatives of the Green's function $ G(p,p) $ at the diagonal can be estimated using the results of the last section. The key is that all relevant terms are given by $ \L $:
\begin{prop} \label{prop:partials} 
Asymptotics of the partial derivatives of the Green's functions. For $ p \in X_0 $, let $ r_{p,\partial M} $ be the Euclidean distance in $ \mathbb{R}^2 $ of $ p $ to $ \partial M $. Then:
\begin{align*}
D_x^2 (\Delta^i)^{-1} (p,p) &=  D_x^2 \L(p,p) +\frac{O(\epsilon^4)}{d^3} \\
D_x (\Delta^i)^{-1}) (p,p) &= O(\epsilon)
\end{align*}
and similar expressions for the $ y $-direction. When $ p$ is near the boundary, $ d_{p, \partial M} \leq \delta_{\partial M } $, then we have a different estimate including the image charges. Let $ p \in B_i $. Then:
\begin{align*}
D_x^2 (\Delta^i)^{-1} (p,p) &=  D_x^2 \L_(p,p) - D_x^2 \L_{R_i}( p, p) + O(\epsilon^2) \\
D_x(\Delta^i)^{-1} (p,p) &=  D_x \L_{R_i}( p, p) + O(\epsilon)
\end{align*}
Finally when $ p \in C_{ij} $, we have:
\begin{align*}
D_x^2(\Delta^i)^{-1}(p,p) &=   D_x^2 \L_(p,p) - D_x^2 \L_{R_i}( p, p) - D_x^2 \L_{R_j}( p,p)+ D_x^2 \L_{R_{ij}}(p,p)+ O(\epsilon^2) \\
D_x (\Delta^i)^{-1}(p,p) &= -D_x \L_{R_i}( p, p) - D_x \L_{R_j}( p,p)+ D_x \L_{R_{ij}}(p,p)+ O(\epsilon)
\end{align*}
and similar expressions for the $y$-direction.
\end{prop}

\begin{proof}
Recall that from (\ref{eq:series}) that:
\begin{align*}
(\Delta^i)^{-1} = G^i + G^i R^i + G^i R^i R^i  \cdots 
\end{align*}
We estimate $ D_x^2 (\Delta^i)^{-1}(p,p) $ and $ D_x (\Delta^i)^{-1}(p,p) $ by first calculating $ D_x^2 G^i(p,p) $ and $ D_x G(p,p) $. Then we show that the higher order terms only contribute at $ O(\e^2) $.

As before, we assume that $ p \not \in B $. Recall that $ G^i(p_1,p_2) = G(p_1,p_2) - \L(p_1,p_2) + \L(p_1,p_2) $. By direct computation using the expansion (\ref{eq:gdist}) for $ G- L $, we find:
\begin{align*}
 D_x^2 G(p,p) - D_x^2 L(p,p) + D_x^2 \L(p,p)  = D_x^2 \L(p,p) + O(\e^2) 
D_x G(p,p) - D_x L(p,p) + D_x \L(p,p) =  O(\e)
\end{align*} 
and similar expressions for the $ y $-direction.

We turn now to the second order term:
\begin{align*}
 D_x^2 G^i R^i (p_1,p_1) = \sum_{p_2 \in X^i_0} D_x^2 G^i(p_1,p_2) R(p_2,p_1) 
\end{align*}
Fix $ r $ to be the (euclidean) distance to the boundary, ie $ r =  \inf \{ |p - p_2 | \; \big| \; p_2 \in \partial M \} $. We break the sum:
\begin{align}
\sum_{p_2 \in X^i_0} D_x^2 G^i(p_1,p_2) R(p_2,p_1)  &= D_x^2 G^i(p_1,p_1) R^i(p_1,p_1) \label{eq:cent} \\ & + \sum_{0 < |p_2 - p_1| \leq r }  D_x^2 G^i(p_2,p_1) \; R^i(p_1,p_2) \label{eq:ins} \\ &+\sum_{ |p_2 - p_1| > r } D_x^2 G(p_2,p_1) R^i(p_1,p_2)  \label{eq:out}
\end{align}
and estimate each term seperately. 

We begin with the region  $ | p_2 - p_1 | > r $. We can estimate $ R^i(p_1,p_2) $ as before with equation (\ref{eq:p1t3}):
\begin{align*}
| R^i(p_1,p_2) | = \frac{O\left(\epsilon^4\right) \;}{ | p_2 - p_1 |^3 } 
\end{align*}
To estimate $ D_x^2 G^i(p_2,p_1) $, we use (\ref{eq:p1t2}) and (\ref{eq:logasymp}) to find:
\begin{align*}
D_x^2 G^i(p_2,p_1) = \frac{ O(\epsilon^2) }{ |p_2 - p_1|^2} 
\end{align*}
Thus we find:
\begin{align} \label{eq:f1}
\sum_{ |p_2 - p_1| > r } D_x^2 G(p_2,p_1) R^i(p_1,p_2) &= \sum_{ |p - p_1| > r } \frac{O(\epsilon^6)}{ | p_2 - p_1|^5 } \notag \\
&= \frac{O(\epsilon^4)}{r^3}
\end{align}

Next, we turn to the region $ 0 < | p_2 - p_1| < r $. This region requires more work because there is considerable cancellation. We first find a more refined estimate for $ R^i(p_1,p_2) $ and show that largest terms of $ O(\e) $ in $ R^i(p_1,p_2) $ are antisymmetric about $ p_1 $ (ie. odd as a function of $(x,y) = p_2 - p_1 $ ). We show that $D_x^2 G^i(p_1,p_2) $, on the other hand is symmetric about $ p_1 $. Thus, the sum (\ref{eq:ins}) vanishes at $ O(\epsilon) $. 

As before, we write $ R^i(p_1,p_2) $ as in (\ref{eq:p1t1}). Considering the term:
\begin{align} \label{eq:b1}
| \Delta^i G(p_1,p_2) + \Delta^i L(p_1,p_2) - \Delta G(p_1,p_2) - \Delta L(p_1,p_2) | 
\end{align}
we can use the expansion (\ref{eq:p1t1}) and expand for small $ (x,y) $:
\begin{align*}
 1 + \frac{ \frac{1}{2} \partial_x  g_{xx} \; x^3 +   \frac{1}{2} \partial_y  g_{xx} \; x^2 y+  \frac{1}{2} \partial_x  g_{yy} \; x y^2+ \frac{1}{2} \partial_y  g_{yy} \; y^3 }{g_{xx} x^2 + g_{yy} y^2}  + \cdots
\end{align*}
The first term is clearly odd in $ (x,y) $, and the fourth partial derivatives of the higher order terms are only $ O(1 / (p_2-p_1)^2) $ as $ p_2 \rightarrow p_1 $. Consequently, using (\ref{eq:lapasym}) again, we write:
\begin{align} \label{eq:p2t1}
| \Delta^i G(p_1,p_2) + \Delta^i L(p_1,p_2) - \Delta G(p_1,p_2) - \Delta L(p_1,p_2) | = O(\epsilon^4) f(p_1,p_2) + \frac{ O(\epsilon^4) }{ |p_2 - p_1|^2}
\end{align}
for a function $ f $ that is anti symmetric with about $ p_1 $ and $ f(p_1,p_2) = O(\frac{1}{|p_1 - p_2|^3 } $ as $ p_2 \rightarrow p_1 $.  Next, compute a similar estimate for $ \Delta^i \mathfrak{L} $. Since both $ \mathfrak{L}(p_1,p_2) $ and $ L(p_1,p_2) $ are even about $ p_1 $ it follows
\begin{align*}
D_x^2 \mathfrak{L}(p_1,p_2) - \partial_x^2 L(p_1,p_2) \\
D_y^2 \mathfrak{L}(p_1,p_2) - \partial_y^2 L(p_1,p_2) 
\end{align*}
are even. Since, as before we have $ g(p_1) D_x^2 + g(p_1) D_y^2 \L = 0 $ we find using (\ref{eq:lapasymp}) that
\begin{align} \label{eq:p2t2}
\Delta^i \mathfrak{L}(p_1,p_2) - \Delta L(p_1,p_2) = \epsilon^4 f_2(x,y) + \frac{O(\epsilon^2)}{ | p_2 - p_1 |}
\end{align}
for a function $ f_2 $ that is odd and $ f_2(x,y) = O(\frac{1}{ | p_1 - p_2| ^3 }) $. Putting together equations (\ref{eq:p2t1}) and (\ref{eq:p2t2}), we find:
\begin{align} \label{eq:p2t3}
R^i(p_1,p_2) = \e^4 f_1(p_1,p_2)+ \frac{O(\epsilon^4)}{|p_2 - p_1|^2}
\end{align}
where $ f_1(p_1,p_2) $ is an odd function, bounded by $ \frac{1}{ | p_1 - p_2 |^3} $. 

Next we estimate $ D_x^2 G^i(p_2,p_1) $. Assume $ p_2 \not\in B $.  First we have:
\begin{align*}
D_x^2 G(p_1,p_2) - D_x^2 L(p_1,p_2) = \frac{O(\e^2)}{ |p_1-p_2| }
\end{align*}
Turning to $ D_x^2 \mathfrak{L}(p_2,p_1) $, recall that:
\begin{align*}
D_x^2 \mathfrak{L}(p_2,p_1)  = D_x^2 \mathfrak{L}_{g(p_2)} (x/\e, y \e)
\end{align*}
We can expand out the dependence on the metric as:
\begin{align*}
 D_x^2 \mathfrak{L}_{g(p_2)} (x/\e, y / \e) = D_x^2 \mathfrak{L}_{g(p_1)} (p_2,p_1) + \frac{O(\e^2)}{|p_1-p_2|}
\end{align*}
Notice however that $ D_x^2 \L $ is an even function. If $ p_2 \in B $, there are additional image charges are bounded by $  \frac{O(\e^2)}{r^2} $. In either case, we can write:
\begin{align} \label{eq:p2t4}
D_x^2 G^i(p_1,p_2) = \e^2 f_2(p_1,p_2) + \frac{O(\e^2)}{|p_1-p_2|}
\end{align}
for a function $ f_3 $ even about $ p_1 $.

Putting the estimates (\ref{eq:p2t4}) and (\ref{eq:p2t3}) into (\ref{eq:ins}), we have finally:
\begin{align*}
\sum_{0 < |p_2 - p_1| \leq r }  D_x^2 G^i(p_2,p_1) = \sum_{0 < |p_2 - p_1| \leq r } \e^6 f_2(p_1,p_2) f_1(p_1,p_2) + \frac{O(\e^6)}{|p_2-p_1|} = O(\e^5)
\end{align*}
where the first term vanishes due to the even-odd.

The contribution of the higher order terms in (\ref{eq:series}) can be estimated using (\ref{eq:p1t3}) and we omit the details. 
\end{proof}

Returning to equation (\ref{eq:trace}), note first that there are $ O(\frac{1}{\epsilon^2}) $ vertices as $ \epsilon \rightarrow 0 $, so that any summand of $ O(\epsilon^2) $ in  will contribute at most a constant term in the asymptotic expansion, which we neglect. Recall that $ B$, $ C $. Thus using Proposition (\ref{prop:partials}), we drop irrelevant terms, and rearrange equation (\ref{eq:trace}) as:
\begin{align} \label{eq:trsum}
 \text{Tr}\left( (\Delta^i_L)^{-1} \delta \Delta_L^i\right) = \sum_{p \in X \backslash L } \text{[bulk]} + \sum_{p \in B}\text{[boundary]} + \sum_{p \in C} \text{[corner]} 
\end{align}
where the summands are:
\begin{align} 
\text{[bulk]} &=  -\frac{\delta\lambda}{\lambda} +  \frac{1}{\lambda} \left( \delta \omega \; D_x^2 - \frac{ \delta \omega}{\omega^2} \; D_y^2 \right) \L(p_1,p_1) \label{eq:sum1} \\
\text{[boundary]} &= \left( \frac{\delta \omega}{\lambda} D_x^2  - \frac{\delta \omega}{\lambda \omega^2}  D_y^2 + \e \partial_x \delta \omega D_x  - \e \delta  \left( \frac{\partial_y \omega}{ \omega^2 }\right) D_y \right) \L_{R_i}(p_1,p_2) \label{eq:sum2} \\
\text{[corner]} &= \left( \frac{\delta \omega}{\lambda} D_x^2  - \frac{\delta \omega}{\lambda \omega^2}  D_y^2 + \e \partial_x \delta \omega D_x  - \e \delta  \left( \frac{\partial_y \omega}{ \omega^2 }\right) D_y \right) \L_{R_ij}(p_1,p_2)  \label{eq:sum3}
\end{align}
we compute each of these sums in the following subsections. To compute these sums, we utilize the Euler-MacLaurin formula, which we write as:
\begin{prop} \label{prop:el}
Let $ f \in C^\infty \left([0,1]\right) $, and $ N \in \mathbb{N} $, then:
\begin{align*}
\sum_{i = 1}^{N-1} f( i / N) &= N \int_0^1 f(x) dx - \frac{1}{2}( f(0) + f(1) ) + O\left(\frac{1}{N}\right)
\end{align*}
Let $ S = [0,1] \times [0,1] $ be the unit square and $ g \in C^\infty(S) $. Then
\begin{align*}
\sum_{i = 1}^{N-1} \sum_{j = 1}^{M-1} g(i/N, j/N) &= N^2 \int_S g(x,y)\;  dx \; dy -  \frac{N}{2}\int_{\partial S} g(s) ds+ O(1) 
\end{align*}
\end{prop}

\subsubsection{Bulk Term}
The bulk sum is:
\begin{align*}
&\sum_{p \in X\backslash L}  -\frac{\delta\lambda}{\lambda} + \frac{1}{\lambda} \left(\delta \omega \frac{2 \lambda}{  \pi \omega } \tan^{-1} \left( \omega \right) - \frac{\delta \omega}{\omega^2}  \frac{2 \lambda \omega}{  \pi } \tan^{-1} \left( \frac{1}{\omega} \right) \right) \\
 =& \sum_{p \in X \backslash L} -\frac{\delta \lambda}{ \lambda} + \frac{2 \; \delta \omega}{\pi \; \omega}\left( \tan(\omega) -  \tan^{-1}( \omega^{-1} ) \right)
\end{align*}
Using the Euler-MacLaurin formula, we find that the bulk contribution to the asymptotic determinant is:
\begin{equation} \label{eq:sum1b}
\begin{aligned}
\frac{1}{\epsilon^2} &\int_M  -\frac{\delta \lambda}{ \lambda} + \frac{2 \; \delta \omega}{\pi \; \omega} \left( \tan(\omega) -  \tan^{-1}( \omega^{-1} ) \right) \; dx \; dy  \\ &- \frac{1}{2 \epsilon} \int_{\partial M} -\frac{\delta \lambda}{ \lambda} + \frac{2 \; \delta \omega}{\pi \; \omega}\left( \tan(\omega) -  \tan^{-1}( \omega^{-1} ) \right) ds
\end{aligned}
\end{equation}

\subsubsection{Boundary Term}
We will compute the contribution of one boundary region $ B_i $. We first sum in the normal direction of the boundary $ \partial M_i $, and then sum along the boundary using the Euler-Maclaurin formula.  We assume is aligned with the $ y$-axis. We choose coordinates $ (x,y) $ chosen so that $ \partial M_i $ is on the $ y $-axis with the interior of $ M$ in the region $ x> 0 $.

Observe that in the boundary summand (\ref{eq:sum2}) by symmetry the term $ D_y \L_{R_i} $ vanishes. In addition, $ \left( \frac{\omega}{\lambda} D_x^2 + \frac{1}{\lambda \omega} D_y^2 \right) \L_{R_i}  $. So we rewrite (\ref{eq:sum2}) as:
\begin{align} \label{eq:be1}
2 \frac{\delta \omega(p_1)}{\lambda(p_1)} D_x^2 \L_{R_i}(p_1,p_1) + \e \frac{\partial_x \delta \omega(p_1)}{ \lambda(p_1)} D_y \L_{R_i} (p_1,p_1)
\end{align}
We expand the terms in the $ x $-coordinate, normal to the boundary:
\begin{align*}
\frac{\delta \omega(x,y)}{\lambda(x,y)} &=  \frac{2 \delta \omega(0,y)}{\lambda(0,y)} + x \; 2\frac{\partial_x \delta \omega(0,y) \lambda(0,y) - \partial_x \lambda(0,y) }{\lambda(0,y)^2}+ \cdots \\ 
\frac{\partial_y \delta \omega(x,y)}{ \lambda(x,y) } &=  \frac{\partial_x \delta \omega(x,0)}{ \lambda(x,0) } + \cdots
\end{align*}
Recall that $ \L_{R_i}(p_1,p_1)  = \L_{g(x,y)}( 2 x / \e,0) $. We can expand the spatial dependence in the metric $ g(x,y) $ as:
\begin{align}
D_x^2 \L_{g(x,y)}(2 x / \e ,0) = D_x^2 \L_{g(x,0)} (2 x/ \e) - \frac{1}{8 \pi } \frac{\e \; \partial_x \lambda(x,0)}{x} + \frac{O(\e^3)}{x^2}
\end{align}
This expression can be justified with using the definition of $\L$ and the equation (\ref{eq:logdef}), but it is easy to guess from the asymptotic expression (\ref{eq:logasymp}), which we rewrite in terms of $ \lambda $ and $ \omega $ as:
\begin{align*}
D_x^2 \L_{g(x,y)}(2 x / \e,0) &=-\frac{\lambda(x,y) \e^2}{ 8 \pi x^2} + \frac{O(\e^4)}{x^4} \\
D_x \L_{g(x,y)}(2 x / \e, 0) &= \frac{\lambda(x,y) \e^2}{4 \pi x} + \frac{O(\e^3)}{x^2}
\end{align*}
Using these expansions in (\ref{eq:be1}), we find that all terms of order $ \frac{O(\e)}{x} $ cancel so that we have for  (\ref{eq:be1}):
\begin{align*}
  \frac{2 \delta \omega(0,y)}{\lambda(0,y)} D_x^2 \L_{g(x,0)} (2 x / \e, 0) + \frac{O(\e^3)}{x^2}
\end{align*}
We now sum normal to the boundary:
\begin{align*}
\sum_{x = \e}^\delta&  \left(  \frac{2 \delta \omega(0,y)}{\lambda(0,y)} D_x^2 \L_{g(x,0)} (2 x / \e, 0) + \frac{O(\e^3)}{x^2}\right) \\
&=  -\frac{\delta \omega(0,y)}{2 \sqrt{ 1+ \omega(0,y)^2}} + \frac{2 \delta \omega(0,y)}{ \pi \omega(y,0) }  \tan^{-1}(\omega)  + O(\e) 
\end{align*}
Finally, we sum along the boundary using the Euler-Maclaurin formula:
\begin{align*}
\frac{1}{\e} \int_{\partial M_i}  \left(-\frac{\delta \omega }{2 \sqrt{ 1+ \omega^2}} + \frac{2}{\pi}\frac{\delta \omega}{\omega}+  \frac{\delta \omega}{ \pi \omega} (\tan^{-1}(\omega)  -\tan^{-1}(\omega^{-1}) \right) dy + O(1) 
\end{align*}
For the other boundary segments oriented along the x-axis, we have the same formula but with $ \omega \rightarrow \frac{1}{\omega} $. We summarise this by writing:
\begin{align}
\delta B_x &= -\frac{\delta \omega}{2 \sqrt{ 1+ \omega^2}} + \frac{\delta \omega}{\omega}  \\
\delta B_y &= \frac{\delta \omega}{2 \omega \sqrt{ 1+ \omega^2}} -\frac{\delta \omega}{\omega}
\end{align}
And writing the boundary contribution as:
\begin{align} \label{eq:sum2b}
\frac{1}{\e} \int_{\partial M} \delta B_{\perp} \; + \frac{\delta \omega}{ \pi \omega} (\tan^{-1}(\omega)  -\tan^{-1}(\omega^{-1})  ds + O(1)
\end{align}

\subsubsection{Corner Term} Using methods very similar to the previous subsection, it is straightforward to check that the corner terms (\ref{eq:sum3}) do not contribute to the asymptotic expansion. We omit the details.

\subsection{Determinant}
We now compile together (\ref{eq:sum1b}) and (\ref{eq:sum2b}) into equation (\ref{eq:trsum}). The boundary term in (\ref{eq:sum1}) cancels with part of (\ref{eq:sum2}). So we have in total:
\begin{align} \label{eq:vareq}
\delta \log \det \Delta = \frac{1}{\e^2} \int_M \delta F dx \; dy \; +    \frac{1}{\e} \int_{\partial M} \delta B_{\perp} ds  + O(1) 
\end{align}
where:
\begin{align*}
\delta F &=  -\frac{\delta \lambda}{ \lambda} + \frac{2 \; \delta \omega}{\pi \; \omega} \left( \tan(\omega) -  \tan^{-1}( \omega^{-1} ) \right) \\
\delta B_x &= -\frac{\delta \omega}{2 \sqrt{ 1+ \omega^2}} + \frac{\delta \omega}{\omega}  - \frac{\delta \lambda}{\lambda} \\
\delta B_y &= \frac{\delta \omega}{2 \omega \sqrt{ 1+ \omega^2}} - \frac{\delta \omega}{\omega} - \frac{\delta \lambda}{\lambda}
\end{align*}
To compute the determinant, we integrate (\ref{eq:vareq}). The constant of integration is fixed when $ \lambda = 1 $ and $ \omega = 1 $ by the expansion (\ref{eq:eq2}). We find:
\begin{align*}
\log \det \Delta =  \frac{1}{\e^2} \int_M  F dx \; dy \; +    \frac{1}{\e} \int_{\partial M} B_{\perp} ds  + O(1) 
\end{align*}
where
\begin{align*}
F &= \log(\lambda) + \frac{2}{\pi}  \Im \left( \text{\emph{Li}}_2 \left(i \omega \right) + \text{\emph{Li}}_2 \left(i \omega^{-1} \right) \right) \notag \\
B_x &= -\log(1 + \sqrt{1+ \omega^{-2}} ) - \log(\lambda) \\
B_y &= -\log(1 + \sqrt{1+ \omega^2} ) - \log(\lambda) \\ \notag
\end{align*}
Substituting $ \lambda = \sqrt{ g_{xx} g_{yy}} $ and $ \omega = \sqrt{ g_{yy} / g_{xx}} $ yields (\ref{eq:densities}). 

\section{Conclusions} \label{sec:sec6}
There are many areas where we believe this work can be improved or generalized. We conclude by outlining some of these areas.

In this paper, we have focused on the square lattice on which one can approximate diagonal metrics. To study more general metrics, one could for example add the diagonal edges of the square lattice. Our methods do not generalize easily to this case - the difficulty is in calculating boundary contributions. In fact,the case where the metric is constant but nondiagonal has not been computed in the literature. One could also consider massive Laplacian operators of the form $ \Delta + m^2 $ or even more generally $ \Delta+ V(x) $ for some potential function $ V \in C^\infty(M) $. In this case, we expect the asymptotic expansion to the be same up to logarithmic terms, although we have not tried to prove this in this paper. For results in this direction, see \cite{chaum}.

Another avenue for generalization is to study the combinatorial Laplacian on non-rectangular regions, for instance, rectilinear regions in $ \mathbb{Z}^2 $. We conjecture that the asymptotic expansion for this case is the obvious generalization of (\ref{eq:eq3}). An approach to this would be to derive gluing formulas for the asymptotic expansions. Similar gluing formulas were derived in \cite{Kenyon2000} by exploiting the relation to the spanning tree model. Alternatively, gluing formulas for determinants can be written in terms of the discrete Dirichlet-to-Neumann operator, which could be studied in the asymptotic limit using estimates for the discrete Green's function given in this paper.

In this paper, we have focused on the laplacian and have not studied the discrete dirac operator, and dimer model on graphs with Riemannian structure. We plan to study this in a future publication. Existing results already give the leading order contribution in the asymptotic partition function and also the average height function \cite{Kenyon2009}, but we believe it is possible to strengthen these results for the dimer models defined in this paper.

Finally, as we have previously mentioned, we have not attempted to investigate the constant term in the asymptotic expansion of the determinant. It is clear that the constant term depends on details of the combinatorial Laplacian not fixed by requiring its consistency. However, for the pure graph laplacian on the square lattice, corresponding to a Euclidean metric, it is known that the constant term in the expansion is equal to the $ \zeta$-regularized determinant of the analytic laplacian. We question if there is a natural condition on combinatorial Laplacians for which the constant term in the asymptotic expansion of the determinant is equal to the $ \zeta $-regularized determinant. Identifying such a condition would be a step towards defining a free Gaussian field theory as the regularized limit of discrete Gaussian field theories.

\end{document}